\documentclass[a4paper,11pt]{article}

\usepackage{authblk}

\usepackage{amsmath}
\usepackage{amsthm}
\usepackage{amssymb}

\usepackage{algorithmic}
\usepackage{algorithm}

\usepackage{graphicx}

\usepackage[left=1.0in,top=1.0in,right=1.0in,bottom=1.0in,nohead]{geometry}

\usepackage{wrapfig}

\theoremstyle{definition}
\newtheorem{definition}{Definition}

\newtheorem{theorem}{Theorem}

\newtheorem{lemma}[theorem]{Lemma}

\graphicspath{{../figures/}}

\newcommand{\GG}{\mathcal G}

\newcommand{\EE}{\mathcal E}
\newcommand{\LL}{\mathcal L}
\newcommand{\HH}{\mathcal H}
\newcommand{\QQ}{\mathcal Q}

\newcommand{\DD}{\mathcal D}
\newcommand{\TT}{\mathcal T}

\newcommand{\RR}{\mathcal R}
\newcommand{\UU}{\mathcal U}
\newcommand{\AAA}{\mathcal A}
\newcommand{\SSS}{\mathcal S}

\newcommand{\CEIL}[1]{\left\lceil #1\right\rceil}
\newcommand{\BIGP}[1]{\left( #1\right)}
\newcommand{\BIGLR}[3]{\left#1#3\right#2}

\long\def\longdelete#1{}

\title{\bf Capacitated Domination: Constant Factor Approximations for Planar Graphs
\thanks{This work was supported in part by the National Science Council, Taipei 10622, Taiwan, under
the Grants NSC99-2911-I-002-055-2, NSC98-2221-E-001-007-MY3, and NSC98-2221-E-001-008-MY3.}
}

\author{Mong-Jen Kao}
\author{D.T. Lee}
\affil{Department of Computer Science and Information Engineering, \\
National Taiwan University, Taiwan. \\
Email: d97021@csie.ntu.edu.tw, dtlee@iis.sinica.edu.tw}

\date{}

\begin{document}

\maketitle

\begin{abstract}
We consider the capacitated domination problem, which
models a service-requirement assigning scenario and which is also a
generalization of the dominating set problem. In this problem, we are given a graph
with three parameters defined on the vertex set, which are cost, capacity, and demand.
The objective of this problem is to compute a demand assignment of least cost,
such that the demand of each vertex is fully-assigned to some of its closed neighbours
without exceeding the amount of capacity they provide.
In this paper, we provide 
the first constant factor approximation for this problem on planar graphs,
based on
a new perspective on the hierarchical structure of
outer-planar graphs. 
We believe that this new perspective and technique can be applied 
to other capacitated covering problems to help tackle vertices of large degrees.
\end{abstract}

\section{Introduction}

For decades, {\em Dominating Set} problem has been one of the most
fundamental and well-known problems in both 
graph theory and combinatorial optimization. Given a graph $G=(V,E)$ and an integer $k$,
{\em Dominating Set} asks for a subset $D\subseteq V$ whose cardinality
does not exceed $k$ such that every vertex in the graph either
belongs to this set or has a neighbour which does. As this problem is
known to be NP-hard, 
approximation algorithms have been proposed in the literature~\cite{174650,DSH82,804034}. 

A series of study on capacitated covering problem was initiated by Guha et al.,~\cite{989540},
which addressed the capacitated vertex cover problem from a scenario of Glycomolecule ID (GMID) placement.
Several follow-up papers have appeared since then, studying both this topic and related 
variations~\cite{1151271,Gandhi:2006:IAA:1740416.1740428,Gandhi:2002:DRB:645413.652158}.
These problems are also closely related to work on the capacitated facility location problem,
which has drawn a lot of attention since 1990s. See~\cite{Chudak:2005:IAA:1047770.1047776,Shmoys:1997:AAF:258533.258600}.


Motivated by a general service-requirement assignment scenario, Kao et al., 
\cite{Kao:2010:AAC:1881195.1881214,springerlink:10.1007/s00453-009-9336-x} considered a generalization of the dominating set problem 
called {\em Capacitated Domination},
which is defined as follows.
%
%
%
%
%
Let $G=(V,E)$ be a graph with three non-negative parameters defined on each vertex $u \in V$,
referred to as the cost, the capacity, and the demand, 
further denoted by $w(u)$, $c(u)$, and $d(u)$, respectively. 
The demand of a vertex stands for the amount of service it requires from its adjacent vertices, 
including the vertex itself, while the capacity of a vertex represents the amount of service 
each 
multiplicity (copy) of that vertex can provide.

By a demand assignment function $f$
we mean a function which maps pairs of vertices to non-negative real numbers.
Intuitively, $f(u,v)$ denotes the amount of demand of $u$ that is assigned to $v$.
We use $N_G(v)$ to denote the set of neighbours of a vertex $v \in V$.

\begin{definition}[feasible demand assignment function] \label{def_feasible_assignment}
A demand assignment function $f$ is said to be feasible if
$\sum_{u \in N_G[v]}f(v,u) \ge d(v)$, for each $v \in V$,
where $N_G[v] = N_G(v) \cup \BIGLR{\{}{\}}{v}$ denotes the neighbours of $v$ unions $v$ itself.
\end{definition}

Given a demand assignment function $f$, the corresponding
capacitated dominating multi-set $\DD(f)$ is defined as follows. For each vertex $v \in V$,
the {\it multiplicity} of $v$ in $\DD(f)$ is defined to be 
$x_f(v) = \CEIL{\frac{\sum_{u \in N_G[v]}f(u,v)}{c(v)}}.$
%
The cost of the assignment function $f$, denoted $w(f)$, is defined to be $w(f)= \sum_{u\in V} w(u)\cdot x_f(u)$.

\begin{definition}[Capacitated Domination Problem]
Given a graph $G=(V,E)$ with cost, capacity, and demand defined on each vertex,
the capacitated domination problem asks for a feasible demand assignment
function $f$ such that $w(f)$ is minimized.
\end{definition}

%



For this problem,
Kao et al.,~\cite{springerlink:10.1007/s00453-009-9336-x}, presented
a $(\Delta+1)$-approximation for general graphs, where $\Delta$ is the maximum vertex degree of the graph,
and a polynomial time approximation scheme for trees, which they proved to be NP-hard. 
In a following work~\cite{Kao:2010:AAC:1881195.1881214}, they provided more approximation algorithms
and complexity results
for this problem. 
%
On the other hand,
Dom et al., \cite{DBLP:conf/iwpec/DomLSV08} 
considered a variation of this problem where the number of multiplicities available at each vertex is limited and
proved the {\it W[1]}-hardness when parameterized by treewidth and solution size.
Cygan et al.,~\cite{springerlink:10.1007/978-3-642-13731-08}, made an attempt toward the exact solution and presented an
$O(1.89^n)$ algorithm when each vertex has unit demand. This result was further improved by Liedloff et al.,~\cite{Liedloff:2010:SCD:1939238.1939250}.


\paragraph{Our Contributions}

We provide the first constant factor approximation algorithms 
for the capacitated domination problem 
on planar graphs. 
This result can be considered a break-through with respect to the pseudo-polynomial 
time approximations
given 
in~\cite{Kao:2010:AAC:1881195.1881214}, which is based on a dynamic programming
on graphs of bounded treewidth.
The approach used in~\cite{Kao:2010:AAC:1881195.1881214} stems from the fact that 
vertices of large degrees will fail most of the techniques that 
transform a pseudo-polynomial time dynamic programming algorithm into approximations,
i.e., the error accumulated at vertices of large degrees could not be bounded.


In this work, we tackle this problem using a new approach.
%
Specifically, we give a new perspective toward the hierarchical structure of outer-planar graphs,
which enables us to further tackle vertices of large degrees.
Then we analyse both the primal and the dual linear programs of this problem to obtain the claimed result.
%
We believe that the approach we provided in this paper can be applied to other capacitated covering
problems to help tackle vertices of large degrees as well.
%
%
%
%

\section{Preliminary} \label{preliminary}

We assume that all the graphs considered in this paper are simple
and undirected. Let $G=(V,E)$ be a graph.
We denote the number of vertices, $|V|$, by $n$.
The set of neighbors of a
vertex $v \in V$ is denoted by $N_G(v) = \{u:(u,v) \in E\}$. The
closed neighborhood of $v \in V$ is denoted by $N_G[v] = N_G(v)
\cup \{v\}$.
We use $deg_G(v)$ and $deg_G[v]$ to denote the cardinality of $N_G(v)$
and $N_G[v]$, respectively.
The subscript $G$ in $N_G[v]$ and $deg_G[v]$ will be omitted when
there is no confusion.

A planar embedding of a graph $G$ is a drawing of $G$
in the plane such that the edges intersect only at their endpoints.
A graph is said to be planar if it has a planar embedding.
An outer-planar graph is a graph which adopts a planar embedding
such that all the vertices lie on a fixed circle, and all the edges are 
straight lines drawn inside the circle.
%
%
For $k \ge 1$, $k$-outerplanar graphs are defined as follows.
A graph is $1$-outerplanar if and only if it is outer-planar.
For $k>1$, a graph is called $k$-outerplanar if it has a planar embedding
such that the removal of the vertices on the unbounded face results
in a $(k-1)$-outerplanar graph.

\begin{figure*}
\centering
\fbox{\begin{minipage}{0.8\textwidth}
\begin{alignat}{2}
& \mathrm{Minimize} \quad \sum_{u \in V}w(u)x(u) \notag \\
& \text{subject to} \notag \\
& \quad \sum_{v \in N[u]}f(u,v) - d(u) \ge 0, && u \in V \notag \\
& \quad c(u)x(u) - \sum_{v \in N[u]}f(v,u) \ge 0, \qquad && u \in V \notag \\
& \quad d(v)x(u)-f(v,u) \ge 0, && v \in N[u], \enskip u \in V \notag \\
& \quad f(u, v) \ge 0,\enskip x(u) \in \mathbb{Z}^+ \cup \{0\}, && u,v \in V \label{ILP_cd}
\end{alignat}
\end{minipage}}
\end{figure*}

An integer linear program (ILP) for capacitated domination is given in $(\ref{ILP_cd})$.
The first inequality ensures the feasibility of the demand assignment function $f$
required in Definition~\ref{def_feasible_assignment}.
In the second inequality, we model the multiplicity function $x$ as defined.
The third constraint, $d(v)x(u)-f(v,u) \ge 0$, which seems unnecessary
in the problem formulation, is required to bound the integrality gap between the optimal solution
of this ILP and that of its relaxation.
To see that this additional constraint does not alter the optimality of any optimal solution, 
we have the following lemma.

\begin{lemma} \label{lemma_redundancy_ILP_additional_constraint}
Let $f$ be an arbitrary optimal demand assignment function. 
We have $d(v)\cdot x_f(u) - f(v,u) \ge 0$ for all $u \in V$ and $v \in N[u]$.
\end{lemma}

\begin{proof}[Proof of Lemma~\ref{lemma_redundancy_ILP_additional_constraint}]
Without loss of generality, we may assume that $d(v) \ge f(v,u)$.
For otherwise, we set $f(v,u)$ to be $d(v)$ and the resulting assignment would be feasible
and the cost can only be better.
If $x_f(u) = 0$, then we have $f(v,u)=0$ by definition, and this inequality holds trivially.
Otherwise, if $x_f(u) \ge 1$, then $d(v)\cdot x_f(u)-f(v,u) \ge f(v,u)\cdot \BIGP{x_f(u)-1}\ge 0$.
\end{proof}

However, without this constraint, the integrality gap can be arbitrarily large.
This is illustrated by the following example.
Let $\alpha > 1$ be an arbitrary constant, 
and $\TT(\alpha)$ be an $n$-vertex star, where each vertex has unit demand and unit cost. 
The capacity of the central vertex is set to be $n$, which is sufficient to cover the demand of the entire graph,
while the capacity of each of remaining $n-1$ petal vertices is set to be $\alpha n$.

\begin{lemma} \label{lemma_necessity_ILP_additional_constraint}
Without the additional constraint $d(v)x(u)-f(v,u) \ge 0$, the integrality gap of the ILP $(\ref{ILP_cd})$
on $\TT(\alpha)$ is 
$\alpha$, where $\alpha > 1$ is an arbitrary constant.
\end{lemma}

\begin{proof}[Proof of Lemma~\ref{lemma_necessity_ILP_additional_constraint}]
The optimal dominating set consists of a single multiplicity of the central vertex with unit cost,
while the optimal fractional solution is formed by spending $\frac{1}{\alpha n}$ multiplicity 
at a petal vertex for each unit demand from the vertices of this graph, 
making an overall cost of $\frac{1}{\alpha}$ and therefore an arbitrarily large integrality gap.
\end{proof}

Indeed, with the additional constraint applied, we can refrain from
unreasonably assigning a small amount of demand to any vertex in any fractional solution.
Take a petal vertex, say $v$, from $\TT(\alpha)$ as example,
given that $d(v)=1$ and $f(v,v)=1$, this constraint would force $x(v)$ to be at least $1$,
which prevents the aforementioned situation from being optimal.


For the rest of this paper, for any graph $G$, we denote the optimal values to the integer linear program $(\ref{ILP_cd})$ and
to its relaxation by $OPT(G)$ and $OPT_f(G)$, respectively.
Note that $OPT_f(G) \le OPT(G)$.


\section{Constant Approximation for Outer-planar Graphs} \label{section_outerplanar}

Without loss of generality, we assume that the graphs are connected.
Otherwise we simply apply the algorithm to each of the connected component separately.
%
%
In the following, we first classify the outer-planar graphs into a class of graphs called
general-ladders and show
how the corresponding general-ladder representation can be extracted in $O(n\log^3n)$ time
in \S\ref{subsection_structure}.
Then we consider in \S\ref{subsection_reduction} and \S\ref{section_greedy_charging} 
both the primal and the dual programs of the relaxation of $(\ref{ILP_cd})$
to further reduce a given general-ladder and obtain a constant factor approximation.
We analyse the algorithm in \S\ref{section_overall_analysis} and 
extend our result to planar graphs in \S\ref{section_extension_planar}.


\begin{figure}[h]
\centering
\includegraphics[scale=0.8]{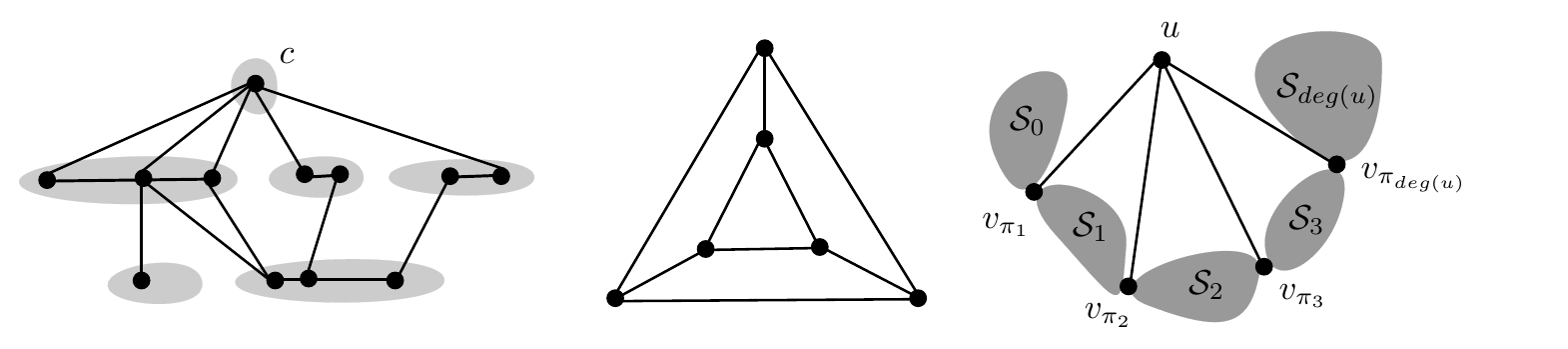}
\caption{(a) A general-ladder with anchor $c$. 
(b) A 2-outerplanar graph which fails to be a general-ladder.
(c) The subdivision formed by a vertex $u$ in an outer-planar embedding.}
\label{figure_ladders_new}
\end{figure}

\subsection{The Structure} \label{subsection_structure}

First we define the notation which we will use later on.
By a total order of a set we mean that each pair of elements in the set can be compared,
and therefore an ascending order of the elements is well-defined.
Let $P = \left( v_1, v_2, \ldots, v_k\right)$ be a path. We say that $P$ is an \emph{ordered path}
if a total order $v_1 \prec v_2 \prec \ldots \prec v_k$ or 
$v_k \prec v_{k-1} \prec \ldots \prec v_1$ is defined on the set of vertices.

\begin{definition}[General-Ladder]
A graph $G = (V,E)$ is said to be a general-ladder if 
a total order on the set of vertices is defined, and
$G$ is composed of
a set of layers $\{\LL_1, \LL_2, \ldots, \LL_k\}$, where each layer is 
a collection of subpaths of an ordered path 
such that the following holds.
The top layer, $\LL_1$, consists of a single vertex, which is referred to as the anchor, and
for each $1< j < k$ and $u,v \in \LL_j$, we have
(1) $N[u] \subseteq \LL_{j-1} \cup \LL_j \cup \LL_{j+1}$, and
(2) $u \prec v$ implies $\max_{p \in N[u] \cap \LL_{j+1}}p \preccurlyeq \min_{q \in N[v] \cap \LL_{j+1}}q$.
%
\end{definition}

Note that
each layer in a general-ladder consists of a set of ordered paths
which are possibly connected only to vertices in the neighbouring layers. 
See Fig.~\ref{figure_ladders_new}~(a).
Although the definition of general-ladders captures the essence and simplicity of an ordered hierarchical structure, there are
planar graphs which fall outside this framework.
See also Fig.~\ref{figure_ladders_new}~(b).

In the following, we state and argue that every outerplanar graph meets the requirements of a general-ladder.
We assume that an outer-planar embedding for any outer-planar graph is given as well. 
Otherwise we apply the $O(n\log^3n)$ algorithm
provided by Bose~\cite{Bose97onembedding} to compute such an embedding.
%



%
Let $G=(V,E)$ be an outer-planar graph, $u \in V$ be an arbitrary vertex, and 
$\EE$ be an outer-planar embedding of $G$.
We fix $u$ to be the smallest element and define a total order on the vertices of $G$ 
according to their orders of appearances on the outer face of $\EE$ in a counter-clockwise order.
For convenience, we label the vertices such that $u=v_1$ and $v_1 \prec v_2 \prec v_3 \prec \ldots \prec v_n$.


Let $N(u) = \BIGLR{\{}{\}}{v_{\pi_1}, v_{\pi_2}, \ldots, v_{\pi_{deg(u)}}}$ denote the neighbours of $u$ such that
$v_{\pi_1} \prec v_{\pi_2} \prec \ldots \prec v_{\pi_{deg(u)}}$.
$N(u)$ divides the set of vertices except $u$ into $deg(u)+1$ subsets, namely, 
$\SSS_0 = \{v_2, v_3, \ldots, v_{\pi_1}\}$,
$\SSS_i = \{v_{\pi_i}, v_{\pi_i+1}, \ldots, v_{\pi_{i+1}}\}$ for $1\le i< deg(u)$, and 
$\SSS_{deg(u)} = \{v_{\pi_{deg(u)}}, v_{\pi_{deg(u)}+1}, \ldots, v_n\}$.
See Fig.~\ref{figure_ladders_new}~(c) for an illustration.

%
\begin{figure}[h]
\centering \fbox{\begin{minipage}{0.4\textwidth}
\centering
\includegraphics[scale=0.8]{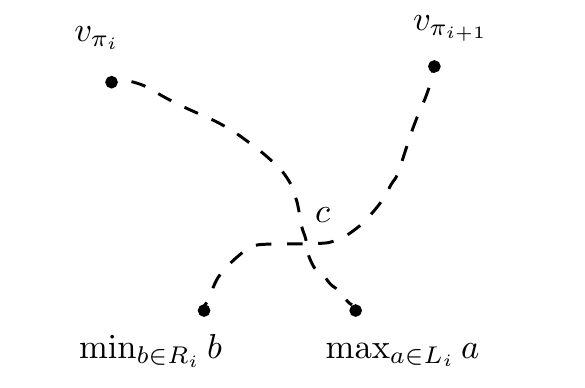}
\caption{A contradiction led by $\min_{b \in R_i}b \prec \max_{a \in L_i}a$.}
\label{figure_subset_partition_conti}
\end{minipage}}
\end{figure}

For any $0<i<j<deg(u)$, $v_{\pi_{i-1}} \prec p \prec v_{\pi_i}$, and $v_{\pi_{j-1}} \prec q \prec v_{\pi_j}$,
there is no edge connecting $p$ and $q$. 
Otherwise it will result in a crossing with the edge $(u, v_{\pi_i})$,
contradicting to the fact that $\EE$ is a planar embedding.

For $1 \le i < deg(u)$, 
we partition $\SSS_i$ into two sets $L_i$ and $R_i$ as follows.
Let $d_{\SSS_i}$ denote the distance function defined on the induced subgraph of $\SSS_i$.
Let $L_i = \BIGLR{\{}{\}}{v: v\in \SSS_i, d_{\SSS_i}\BIGP{v_{\pi_i}, v} \le d_{\SSS_i}\BIGP{v, v_{\pi_{i+1}}}}$ and
$R_i = \SSS_i \backslash L_i$. 


\begin{lemma} \label{lemma_structure_outerplanar_subset_partition}
We have $\max_{a \in L_i}a \prec \min_{b \in R_i}b$ for all $1\le i<deg(u)$.
\end{lemma}

\begin{proof}[Proof of Lemma~\ref{lemma_structure_outerplanar_subset_partition}]
For convenience, let $a = \max_{a \in L_i}a$ and $b = \min_{b \in R_i}b$.
Since $L_i \cap R_i = \phi$, we have $a \neq b$.
Assume that $b \prec a$.
Recall that in an outer-planar embedding, the vertices are placed on a circle and the edges are drawn
as straight lines.
Since $\EE$ is an outer-planar embedding, the shortest path from $a$ to
$v_{\pi_i}$ must intersect with the shortest path from $b$ to $v_{\pi_{i+1}}$.
Let $c$ be the vertex for which the two paths meet. Since $L_i$ and $R_i$ form a partition of $\SSS_i$,
either $c\in L_i$ or $c\in R_i$. 

If $c \in L_i$, then $d_{\SSS_i}\BIGP{c,v_{\pi_i}} \le d_{\SSS_i}\BIGP{c,v_{\pi_{i+1}}}$ by definition,
which implies that $d_{\SSS_i}\BIGP{b, v_{\pi_i}} \le d_{\SSS_i}\BIGP{b, v_{\pi_{i+1}}}$,
a contradiction to the fact that $b \in R_i$.
On the other hand, if $c \in R_i$, then $d_{\SSS_i}\BIGP{c,v_{\pi_i}} > d_{\SSS_i}\BIGP{c,v_{\pi_{i+1}}}$,
and we have $d_{\SSS_i}\BIGP{a, v_{\pi_i}} > d_{\SSS_i}\BIGP{a, v_{\pi_{i+1}}}$, a contradiction to the fact $a \in L_i$.
In both cases, we have a contradiction. Therefore we have $a \prec b$.
\end{proof}

\begin{figure}[t]
\centering \fbox{\begin{minipage}{0.5\textwidth}
\centering
\includegraphics[scale=0.8]{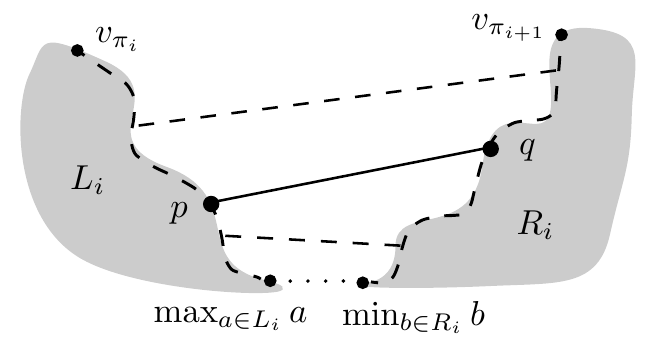}
\caption{Partition of $\SSS_i$ into $L_i$ and $R_i$.}
\label{figure_subset_partition}
\end{minipage}}
\end{figure}
%

Let $\ell(v) \equiv d_G(u,v)$ and
$\ell_i(v) \equiv \min\BIGLR{\{}{\}}{d_{\SSS_i}(v_{\pi_{(i)}}, v), d_{\SSS_i}(v_{\pi_{(i+1)}}, v)}$, 
for any $1\le i<deg(u)$ and $v \in \SSS_i$.
Observe that $\ell(v) = \ell_i(v)+1$, for any $1\le i<deg(u)$ and $v \in \SSS_i$.
Now consider the set of the edges connecting $L_i$ and $R_i$.
Note that, this is exactly the set of edges connecting vertices on the shortest path between
$v_{\pi_i}$ and $\max_{a \in L_i}a$ and vertices on the shortest path between
$v_{\pi_{i+1}}$ and $\min_{b \in R_i}b$.
We have the following lemma, which states that, when the vertices are classified by their distances
to $u$, these edges can only connect vertices between neighbouring sets and do not form any crossing.
See also Fig.~\ref{figure_subset_partition}.

\begin{lemma} \label{lemma_structure_outerplanar_subset_joining}
For any edge $(p,q)$, $p\in L_i$, $q\in R_i$, connecting $L_i$ and $R_i$, we have
\begin{itemize}
\item $\BIGLR{|}{|}{\ell(p)-\ell(q)} \le 1$, and
\item $\nexists$ edge $(r,s)$, $(r,s) \ne (p,q)$, $r\in L_i$, $s\in R_i$, such that $\ell(r)=\ell(q)$ and $\ell(p)=\ell(s)$.
\end{itemize}
\end{lemma}

\begin{proof}[Proof of Lemma~\ref{lemma_structure_outerplanar_subset_joining}]
The first half of the lemma follows from the definition of $\ell$.
If $\BIGLR{|}{|}{\ell(p)-\ell(q)} > 1$, without loss of generality, suppose that $\ell(p)> \ell(q)+1$,
by going through $(p,q)$ then following the shortest path from $q$ to $u$, we find a shorter path for $p$,
which is a contradiction.
The second half follows from the fact that $\EE$ is a planar embedding.
\end{proof}

Below we present our structural lemma,
which states that, when the vertices are classified by their distances
to $u$, these edges can only connect vertices between neighbouring sets and do not form any crossing.
%


\begin{lemma} \label{lemma_structure_outerplanar_general_ladder}
Any outer-planar graph $G=(V,E)$ together with an arbitrary vertex $u \in V$ is a general-ladder anchored at $u$,
where the set of vertices in each layer are classified by their distances to the anchor $u$.
\end{lemma}

\begin{proof}[Proof of Lemma~\ref{lemma_structure_outerplanar_general_ladder}]
We prove by induction on the number of vertices of $G$. 
First, an isolated vertex is a single-layer general-ladder.
For non-trivial graphs,
let $\SSS_0, \SSS_1, \ldots, \SSS_{deg(u)}$ be the subsets defined as above.
By assumption, the induced subgraphs of $\SSS_0$ and $\SSS_{deg(u)}$ are general-ladders 
with anchors $v_{\pi_1}$ and $v_{\pi_{deg(u)}}$, respectively.
Furthermore, the layers are classified by $\ell-1$. That is, vertex $v$ belongs to layer $\ell(v)-1$.
Similarly, 
the induced subgraphs of $L_i$ and $R_i$ are also general-ladders with anchors $v_{\pi_i}$ and $v_{\pi_{i+1}}$
whose layers are classified by $\ell_i$.


Now we argue that these general-ladders can be arranged properly to form
a single general-ladder with anchor $u$ and layers classified by $\ell$.
Since there is no edge connecting $p$ and $q$ for any $p,q$ with
$v_{\pi_{i-1}} \prec p \prec v_{\pi_i}$ and $v_{\pi_{j-1}} \prec q \prec v_{\pi_j}$, $0<i<j<deg(u)$,
we only need to consider the edges connecting vertices between $L_i$ and $R_i$.
By Lemma~\ref{lemma_structure_outerplanar_subset_joining}, when the general-ladders $L_i$
and $R_i$ are hung over $v_{\pi_i}$ and $v_{\pi_{i+1}}$, respectively,
the edges between them connect exactly only vertices from adjacent layers and do not
form any crossing. Therefore, it constitute as a single general-ladder together with $u$ and the lemma follows.
%
\end{proof}

%

%
\paragraph{Extracting the General-ladder}
Let $\GG=(V,E)$ be the input outer-planar graph and $u \in V$ be an arbitrary vertex. 
We identify the corresponding general-ladder as follows.

Compute the shortest distance of each vertex $v \in V$ to $u$, denoted by $\ell(v)$. 
Let $M=\max_{v\in V}\ell(v)$. 
We create $M+1$ empty queues, $layer(0), layer(1), \ldots, layer(M)$, which will be used to maintain the set of layers.
Retrieve an outer-planar embedding of $G$ and traverse the outer face, starting from $u$,
in a counter-clockwise order.
For each vertex $v$ visited, we attach $v$ to the end of $layer(\ell(v))$.

\bigskip

\begin{theorem} \label{theorem_computing_representation}
Given an outer-planar graph $\GG$ and its outer-planar embedding, 
we can compute in linear time a general-ladder representation for $G$.
\end{theorem}

\begin{proof}[Proof of Theorem~\ref{theorem_computing_representation}]
Since the number of edges in a planar graph is linear in the number of vertices,
the shortest-path tree computation takes linear time. 
The traversal of the outer face also takes linear time.
\end{proof}

For the rest of this paper we will denote 
the layers of this particular general-ladder representation by $\LL_0, \LL_1, \ldots, \LL_M$.
The following additional structural property comes from the outer-planarity of $\GG$ and our construction scheme.

\begin{lemma} \label{lemma_ladder_upward_degree_bound}
For any $0<i\le M$ and $v \in \LL_i$, we have $\BIGLR{|}{|}{N(v) \cap \LL_{i-1}} \le 2$.
Moreover, if $v$ has two neighbours in $\LL_i$, say, $v_1$ and $v_2$ with $v_1 \prec v \prec v_2$,
then there is an edge joining $v_1$ (and $v_2$, respectively) and each neighbouring vertex of $v$ in $\LL_{i-1}$
that is smaller (larger) than $v$.
%
%
\end{lemma}

\begin{proof}[Proof of Lemma~\ref{lemma_ladder_upward_degree_bound}]
First, since the layers are classified by the distances to the anchor $u$, if $\BIGLR{|}{|}{N(v) \cap \LL_{i-1}} \ge 3$,
then consider the shortest paths from vertices in $N(v) \cap \LL_{i-1}$ to $u$. At least one vertex would be surrounded by
other two paths, contradicting the fact that $\GG$ is an outer-planar graph.


The second part is obtained from a similar argument. Let $v^\prime \in N(v) \cap \LL_{i-1}$ be a neighbour of $v$
in $\LL_{i-1}$. If $v^\prime$ is not joined to either $v_1$ or $v_2$, then consider the shortest paths
from $u$ to $v_1$, $v^\prime$, and $v_2$, respectively. $v^\prime$ would be a vertex in the interior, which
is a contradiction.
\end{proof}


\paragraph{The Decomposition}
The idea behind this decomposition is to help reduce the dependency between vertices of large degrees
and their neighbours such that further techniques can be applied.
To this end, we tackle the demands of vertices from every three layers separately.


For each $0\le i < 3$, let $\RR_i = \bigcup_{j\ge 0}\LL_{3j+i}$.
Let $\GG_i=(V_i, E_i)$ consist of the induced subgraph of $\RR_i$ and the set of edges connecting vertices in $\RR_i$ to
their neighbours. 
Formally, $V_i = \bigcup_{v \in \RR_i}N[v]$ and $E_i = \bigcup_{v \in \RR_i}\bigcup_{u\in N[v]}e(u,v)$.
In addition, we set $d(v) = 0$ for all $v \in \GG_i \backslash \RR_i$. 
Other parameters remain unchanged.



\begin{lemma} \label{lemma_layers_decomp_quality}
Let $f_i$, $0\le i<3$, be an optimal demand assignment function for $\GG_i$.
The assignment function $f = \sum_{0\le i<3}f_i$ is a $3$-approximation of $\GG$.
\end{lemma}

\begin{proof}[Proof of Lemma~\ref{lemma_layers_decomp_quality}]
First, for any vertex $v \in V$, the demand of $v$ is considered in $\GG_i$ for some $0\le i<3$
and therefore is assigned by the assignment function $f_i$.
Since we take the union of the three assignments, it is a feasible assignment to the entire graph $\GG$.


Since the demand of each vertex in $\GG_i$, $0\le i<3$, is no more than that of in the original graph $G$,
any feasible solution to $G$ will also serve as a feasible solution to $\GG_i$.
Therefore we have $OPT(\GG_j) \le OPT(G)$, for $0\le j<3$, and the lemma follows.
\end{proof}

\subsection{Removing More Edges} \label{subsection_reduction}

We describe an approach to further simplifying the graphs $\GG_i$, for $0\le i<3$.
Given any feasible demand assignment for $\GG_i$,
we can properly reassign the demand of a vertex to a constant number of neighbours while
the increase in terms of fractional cost remains bounded.


For each $v \in \RR_i$, we sort the closed neighbours of $v$ according to their cost
in ascending order 
such that
$w\BIGP{\pi_v(1)} \le w\BIGP{\pi_v(2)} \le \ldots \le w\BIGP{\pi_v(deg[v])}$, where
$\pi_v: \BIGLR{\{}{\}}{1,2,\ldots,deg[v]} \rightarrow N[v]$ is an injective function.
%
For convenience, we set $\pi_v(deg[v]+1) = \phi$.
Suppose that $v \in \LL_\ell$.
We identify the following four vertices.

\begin{itemize}
\item
Let $j_v$, $1\le j_v \le deg[v]$, be the smallest integer such that $c\BIGP{\pi_v(j_v)} > d(v)$.
If $c\BIGP{\pi_v(j_v)} \le d(v)$ for all $1\le j\le deg[v]$, then we let $j_v = deg[v]+1$.

\item Let $k_v$, $1\le k_v < j_v$, be the integer such that and
$w\BIGP{\pi_v(k_v)} / c\BIGP{\pi_v(k_v)}$ is minimized.
$k_v$ is defined only when $j_v > 1$.

\item Let $p_v = \max_{u \in N[v] \cap \LL_{\ell-1}}u$ and 

$q_v = \max_{u \in N[v] \cap \LL_{\ell+1}}u$.
\end{itemize}

\begin{figure}[h]
\centering
\fbox{\begin{minipage}{0.5\textwidth}
\centering
\includegraphics[scale=0.8]{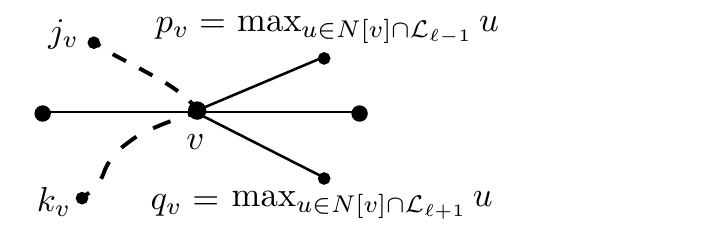}
\caption{Incident edges of a vertex $v \in \LL_\ell$ to be kept.
}
\label{figure_edges_kept}
\end{minipage}}
\end{figure}

Intuitively, $\pi_v(j_v)$ is the first vertex in the sorted list whose capacity is greater than 
$d(v)$,
and $\pi_v(k_v)$ is the vertex with best cost-capacity ratio among the first $j_v-1$ vertices.
$p_v$ and $q_v$ are the rightmost neighbour of $v$ in layer $\LL_{\ell-1}$ and $\LL_{\ell+1}$, respectively.

We will omit the function $\pi_v$ and use $j_v$, $k_v$ to denote $\pi_v(j_v)$, $\pi_v(k_v)$
without confusion.
The reduced graph $\HH_i$ is defined as follows.
Denote the set of neighbours to be disconnected from $v$
by $R(v) = N[v] \backslash \BIGP{\LL_\ell \cup \BIGLR{\{}{\}}{j_v \cup k_v \cup p_v \cup q_v}}$,
and let $\HH_i = \GG_i \backslash \bigcup_{v \in \RR_i}\bigcup_{u \in R(v)}\BIGLR{\{}{\}}{e(u,v)}$.
Roughly speaking, in graph $\HH_i$ we remove the edges which connect vertices in $\RR_i$, say $v$, 
to vertices not in $\RR_i$, except possibly for $j_v$, $k_v$, $p_v$, and $q_v$.
See Fig.~\ref{figure_edges_kept}.
Note that, although our reassigning argument applies to arbitrary graphs,
only when two vertices are unimportant to each other can we remove the edge between them.

\begin{lemma} \label{lemma_bounded_degree_reduction}
In the subgraph $\HH_i$, we have
\begin{itemize}
\item For each $v \notin \RR_i$, at most one incident edge of $v$ which was previously in $\GG_i$
will be removed.

\item For each $v \in \RR_i$, the degree of $v$ in $\HH_i$ is upper-bounded by $6$.

\item $OPT_f(\HH_i) \le 2\cdot OPT_f(\GG_i)$
\end{itemize}
\end{lemma}

\begin{proof}[Proof of Lemma~\ref{lemma_bounded_degree_reduction}]
For the first part,
let $v \notin \RR_i$ be a vertex and denote $\SSS = N[v] \cap \RR_i$ the set of neighbours of $v$
that are in $\RR_i$. 
By the definition of general-ladders, for any $u \in \SSS$, $u \neq \max_{a \in \SSS}a$,
we have either $p_u = v$ or $q_u = v$, since $v$ serves as the rightmost neighbour of $u$.
Therefore, by our approach, only the edge between $v$ and $\max_{a \in \SSS}a$ will
possibly be removed.

For the second part,
for any $v \in \RR_i$, $v$ has at most two neighbours in the same layer, 
since each layer is a subgraph of an ordered path. 
We have removed all the edges connecting $v$ to vertices not in $\LL_k$,
except for at most four vertices, $j_v$, $k_v$, $p_v$, and $q_v$.
Therefore $deg(v) \le 6$.

Now we prove the third part of this lemma.
Let $f_{\GG_i}$ be an optimal demand assignment for $\GG_i$, and $x_{\GG_i}$ be the corresponding
multiplicity function. Note that, from the second and the third inequalities of $(\ref{ILP_cd})$, 
for each $v \in V$ and $u \in N[v]$, we have
\begin{equation} \label{eqn_lp_multiplicity_constraint}
x_{\GG_i}(u) \ge \max\BIGLR{\{}{\}}{\frac{f_{\GG_i}(v,u)}{d(v)}, \frac{f_{\GG_i}(v,u)}{c(u)}}.
\end{equation}


For each $v \in \RR_i$ and $u \in R(v)$ such that $f_{\GG_i}(v,u) \neq 0$, we modify this assignment as follows.
If $\pi^{-1}_v(u) \ge j_v$, then we assign 
it to $j_v$ instead of to $u$.
Otherwise, we assign it to $k_v$.
That is, depending on whether $\pi^{-1}_v(u) \ge j_v$,
we raise either $f_{\GG_i}(v,j_v)$ or $f_{\GG_i}(v,k_v)$ by the amount of $f_{\GG_i}(v,u)$
and then set $f_{\GG_i}(v,u)$ to be zero.
Note that, after this reassignment, 
the modified assignment function $f_{\GG_i}$ will be a feasible assignment for $\HH_i$ as well.


In order to cope with this change, $x_{\GG_i}(j_v)$ or $x_{\GG_i}(k_v)$ might have to be raised as well
until both the second and the third inequalities are valid again.
If $\pi^{-1}_v(u) \ge j_v$, then $x_{\GG_i}(j_v)$ is raised by at most 
$\max\BIGLR{\{}{\}}{f_{\GG_i}(v,u) / d(v), f_{\GG_i}(v,u) / c(j_v)}$, which is equal to $f_{\GG_i}(v,u) / d(v)$, 
since $c(j_v) > d(v)$.
Hence the total cost will be raised by at most 
$$w(j_v)\cdot \frac{f_{\GG_i}(v,u)}{d(v)} \le 
w(j_v) \cdot \max\BIGLR{\{}{\}}{\frac{f_{\GG_i}(v,u)}{d(v)}, \frac{f_{\GG_i}(v,u)}{c(u)}} \le w(u)\cdot x_{\GG_i}(u),$$
by equation~$(\ref{eqn_lp_multiplicity_constraint})$ and the fact that $w(j_v) \le w(u)$.
Similarly, if $\pi^{-1}_v(u) < j_v$, the cost is raised by at most
$$w(k_v)\cdot \max\BIGLR{\{}{\}}{\frac{f_{\GG_i}(v,u)}{d(v)}, \frac{f_{\GG_i}(v,u)}{c(k_v)}}
= w(k_v)\cdot \frac{f_{\GG_i}(v,u)}{c(k_v)} \le w(u)\cdot \frac{f_{\GG_i}(v,u)}{c(u)} \le w(u)\cdot x_{\GG_i}(u),$$
since we have $\frac{w(k_v)}{c(k_v)} \le \frac{w(u)}{c(u)}$ by definition of $k_v$ and equation~$(\ref{eqn_lp_multiplicity_constraint})$.


In both cases, the extra cost required by this specific demand reassignment between $v$ and $u$ 
is bounded by $w(u)\cdot x_{\GG_i}(u)$. 
Since, by Lemma~\ref{lemma_bounded_degree_reduction}, we have at most one such pair for each $u \notin \RR_i$,
the overall cost is at most doubled and this lemma follows.
\end{proof}

%


%
We also remark that, although $OPT_f(\HH_i)$ is bounded in terms of $OPT_f(\GG_i)$, 
an $\alpha$-approximation for $\HH_i$ is not necessarily a $2\alpha$-approximation for $\GG_i$.
That is, having an approximation $\AAA$ with $OPT(\AAA) \le \alpha\cdot OPT(\HH_i)$ does not
imply that $OPT(\AAA) \le 2\alpha\cdot OPT(\GG_i)$,
for $OPT(\HH_i)$ could be strictly larger than $OPT_f(\HH_i)$.
Instead, to obtain our claimed result, an approximation with a stronger bound,
in terms of $OPT_f(\HH_i)$, is desired.

\subsection{Greedy Charging Scheme} \label{section_greedy_charging}

In this section, we show how we can further approximate the optimal solution for the reduced graph $\HH_i$
by a primal-dual charging argument.
We apply a technique from~\cite{springerlink:10.1007/s00453-009-9336-x} to obtain a feasible
solution for the dual program of the relaxation of
$(\ref{ILP_cd})$, 
which is given in $(\ref{ILP_dual_cd})$ and for which
we will further provide a sophisticated analysis on how the cost of each multiplicity we spent can be distributed
to each unit demand it covered in a careful way.
Thanks to the additional structural property provided in Lemma~\ref{lemma_ladder_upward_degree_bound},
we can further tighten the approximation ratio.

\begin{figure*}
\centering \fbox{\begin{minipage}{0.8\textwidth}
\begin{alignat}{2}
& \mathrm{Maximize} \quad \sum_{u\in V}d(u)y_u \notag \\
& \text{subject to} \notag \\
& \quad c(u)z_u + \sum_{v \in N[u]}d(v)g_{u,v} \le w(u), \quad && u\in V \notag \\
& \quad y_u \le z_v + g_{v,u}, && v \in N[u], \enskip u \in V \notag \\
& \quad y_u \ge 0, \enskip z_u \ge 0, \enskip g_{v,u} \ge 0, && v \in N[u], \enskip u \in V \label{ILP_dual_cd}
\end{alignat}
\end{minipage}}
\end{figure*}

We first describe an approach to obtaining a feasible solution to 
$(\ref{ILP_dual_cd})$ and how a corresponding feasible demand assignment can be found.
%
%
Note that any feasible solution to $(\ref{ILP_dual_cd})$ will serve as a lower bound to any feasible solution of $(\ref{ILP_cd})$ 
by the linear program duality.

%
During the process, we will maintain a vertex subset, $V^\phi$, 
which contains the set of vertices with non-zero unassigned demand.
For each $u \in V$, let $d^\phi(u) = \sum_{v \in N[u] \cap V^\phi}d(v)$ denote the amount of unassigned demand
from the closed neighbours of $u$. 
We distinguish between two cases.
If $c(u) < d^\phi(u)$, then we say that $u$ is heavily-loaded. Otherwise, $u$ is lightly-loaded.
During the process, some heavily-loaded vertices might turn into lightly-loaded due to the demand assignments
of its closed neighbours. For each of these vertices, say $v$, we will maintain a vertex subset $D^*(v)$,
which contains the set of unassigned vertices in $N[v] \cap V^\phi$ when $v$ is about to fall into lightly-loaded.
For other vertices, $D^*(v)$ is defined to be an empty set.


Initially, $V^\phi \equiv \{u: u\in \LL_i, d(u) \neq 0\}$ and all the dual variables are set to be zero.
We increase the dual variable $y_u$ simultaneously, for each $u \in V^\phi$. 
To maintain the dual feasibility, as we increase $y_u$, we have to raise either $z_v$ or $g_{v,u}$, for each $v \in N[u]$.
If $v$ is heavily-loaded, then we raise $z_v$. Otherwise, we raise $g_{v,u}$.
Note that, during this process, for each vertex $u$ that has a closed neighbour in $V^\phi$, the left-hand side of the inequality
$c(u)z_u + \sum_{v \in N[u]}d(v)g_{u,v} \le w(u)$ is constantly raising.
As soon as one of the inequalities $c(u)z_u + \sum_{v \in N[u]}d(v)g_{u,v} \le w(u)$ is met with equality (saturated)
for some vertex $u \in V$, we
perform the following operations.

If $u$ is lightly-loaded, we assign all the unassigned demand from $N[u] \cap V^\phi$ to $u$.
In this case, there are still $c(u) - d^\phi(u)$ units of capacity free at $u$.
We assign the unassigned demand from $D^*(u)$, if there is any, to $u$ until either all the demand from $D^*(u)$ is assigned
or all the free capacity in $u$ is used.
On the other hand, if $u$ is heavily-loaded, we mark it as heavy and delay the demand assignment from its closed neighbours.

Then we set $\QQ_u \equiv N[u] \cap V^\phi$
and remove $N[u]$ from $V^\phi$. 
Note that, due to the definition of $d^\phi$, even when $u$ is heavily-loaded, we still update $d^\phi(p)$ for each $p \in V$
with $N[p]\cap N[u] \neq \phi$, if needed, as if the demand
was assigned. During the above operation, some heavily-loaded vertices might turn into lightly-loaded
due to the demand assignments (or simply due to the update of $d^\phi$).
For each of these vertices, say $v$, we set $D^*(v) \equiv N[v] \cap (V^\phi \cup \QQ_u)$.
Intuitively, $D^*(v)$ contains the set of unassigned vertices from $N[v] \cap V^\phi$ when $v$ is about to fall into lightly-loaded.

This process is continued until $V^\phi = \phi$.
For those vertices which are marked as heavy, 
we iterate over them according to their chronological order of being saturated and
assign at this moment all the remaining unassigned demand from their
closed neighbours to them.
A high-level description of this algorithm is given in Fig.~\ref{algorithm_greedy_charging}.

\begin{figure*}[t]
\rule{\linewidth}{0.2mm}
\medskip
{{\sc Algorithm} {\em Greedy-Charging}}

\begin{algorithmic}[1]
\STATE $V^\phi \longleftarrow \{u: u\in V, d(u)\neq 0 \}$, $V^* \longleftarrow V$.
\STATE $d^\phi(u) \longleftarrow \sum_{v \in N[u]}d(v)$, for each $u \in V$.
\STATE $w^\phi(u) \longleftarrow w(u)$, for each $u \in V$.
\STATE Let $\QQ$ be a first-in-first-out queue.
\WHILE{$V^\phi \neq \phi$}
	\STATE $r_v \longleftarrow w^\phi(v) / \min\{c(v), d^\phi(v)\}$, for each $v \in V^*$.
	\STATE $u \longleftarrow arg\min\{r_v: v \in V^*\}$. [$u$ is the next vertex to be saturated.]
	\STATE $w^\phi(v) \longleftarrow w^\phi(v) - w^\phi(u)$, for each $v \in V^*$.
	\STATE
	\IF{$d^\phi(u) \le c(u)$}
		\STATE Assign the demand from $N[u] \cap V^\phi$ to $u$.
		\STATE Assign $c(u) - d^\phi(u)$ amount of unassigned demand from $D^*(u)$, if there is any, to $u$.
	\ELSE
		\STATE Attach $u$ to the queue $\QQ$ and mark $u$ as {\it heavy}.
	\ENDIF
	\STATE Let $\SSS_u \longleftarrow N[u] \cap V^\phi \cup \{u\}$.
	\STATE Remove $N[u]$ from $V^\phi$ and update the corresponding $d^\phi(v)$ for $v \in V$.
	\STATE For each $v \in V^*$ such that $d^\phi(v)=0$, remove $v$ from $V^*$.
	\FORALL{vertex $v$ becomes lightly-loaded in this iteration}
		\STATE $D^*(v) \longleftarrow N[v] \cap (V^\phi \cup \SSS_u)$.
	\ENDFOR
\ENDWHILE
\WHILE{$\QQ \neq \phi$}
	\STATE Extract a vertex from the head of $\QQ$, say $u$.
	\STATE Assign the unassigned demand from $N[u]$ to $u$.
\ENDWHILE
\STATE

\end{algorithmic}
\rule{\linewidth}{0.2mm} 
\caption{The high-level pseudo-code for the primal-dual algorithm.} 
\label{algorithm_greedy_charging}
\vspace{-14pt}
\end{figure*}


Let $f^*: V \times V \rightarrow R^+\cup \{0\}$ denote the resulting demand assignment function, and 
$x^*: V \rightarrow Z^+\cup \{0\}$ denotes the corresponding multiplicity function.
The following lemma bounds the cost of the solution produced by our algorithm.

\begin{lemma} \label{lemma_primal_dual_ladders}
For any $\HH_i$ obtained from a general-ladder $\GG_i$, we have $w(f^*) \le 7\cdot OPT_f(\HH_i)$.
\end{lemma}

\begin{proof}[Proof of Lemma~\ref{lemma_primal_dual_ladders}]
We argue in the following that, for each $u \in V$, the cost resulted by $u$, which is $w(u)\cdot x^*(u)$,
can be distributed to a certain portion of the demands from the vertices in $N[u]$ such that each unit demand, 
say from vertex $v\in N[u]$, receives a charge of at most $7\cdot y_v$.

Let $u \in V$ be a vertex with $x^*(u) > 0$.
If $u$ has been marked as heavy, then by our scheme, we have $g_{u,v} = 0$ and $y_v = z_u$
for all $v \in N[u]$. Therefore $w(u) = c(u)\cdot z_u = c(u)\cdot y_v$, and for each multiplicity of $u$, 
we need $c(u)$ units of demand from $N[u]$.
If $x^*(u) > 1$, then at least $c(u)\cdot (x^*(u)-1)$ units of demand are assigned to $u$, and 
by distributing the cost to them, each unit of demand gets charged at most twice.
If $x^*(u) = 1$, then we charge the cost to any $c(u)$ units of demand that are counted in $d^\phi(u)$ when
$u$ is saturated. Since $u$ is a heavily-loaded vertex, $d^\phi(u) > c(u)$ and there will be sufficient
amount of demand to charge.

On the other hand, if $u$ is lightly-loaded, then $x^*(u) = 1$ and we have two cases.
If $\sum_{v \in N[u]}d(v) \le c(u)$, then $u$ is lightly-loaded in the beginning and we have
$z_u = 0$ and $y_v = g_{u,v}$ for each $v \in N[u]$, which implies
$w(u) = \sum_{v \in N[u]}d(v)\cdot g_{u,v} = \sum_{v \in N[u]}d(v)\cdot y_v$. 
The cost $w(u)$ of $u$ can be distributed to all the demand
from its closed neighbours, each unit demand, say from vertex $v \in N[u]$, gets a charge of $y_v$.

If $\sum_{v \in N[u]}d(v) > c(u)$, then $u$ is heavily-loaded in the beginning and at some point turned into lightly-loaded.
Let $\UU_0 \subseteq D^*(u)$ be the set of vertices whose removal from $V^\phi$ makes this change.
By our scheme, $z_u$ is raised in the beginning and at some point when $d^\phi(u)$ is about to fall under $c(u)$,
we fixed $z_u$ and start raising $g_{u,v}$ for $v \in D^*(u) \backslash \UU_0$.
Note that, we have $y_{u_0} = z_u$ for all $u_0 \in \UU_0$,
$y_v = z_u + g_{u,v}$ for each $v \in D^*(u) \backslash \UU_0$,
and $g_{u,v} = 0$ for $v \in N[u] \backslash D^*(u) \cup \UU_0$.
Let $d^*_{\UU_0} = c(u) - \sum_{v \in D^*(u) \backslash \UU_0}d(v)$.
We have 
$w(u) = c(u)\cdot z_u + \sum_{v \in N[u]}d(v)\cdot g_{u,v} 
= \left(d^*_{\UU_0}+\sum_{v \in D^*(u) \backslash \UU_0}d(v)\right)\cdot z_u + \sum_{v \in D^*(u)}d(v)\cdot g_{u,v}
= d^*_{\UU_0}\cdot z_u + \sum_{v \in D^*(u) \backslash \UU_0}d(v)\cdot (z_u + g_{u,v})
\le \sum_{v \in \UU_0}d_v\cdot y_v + \sum_{v \in D^*(u) \backslash \UU_0}d(v) \cdot y_v$.
Again the cost $w(u)$ of the single multiplicity can be distributed to the demand of vertices in $D^*(u)$.

Finally, for each unit demand, say demand $d$ from vertex $u$, consider the set of vertices $V_d \subseteq N[u]$ that
has charged $d$.
First, by our assigning scheme, $V_d$ consists of at most one heavily-loaded vertex.
If $d$ is assigned to a heavily-loaded vertex, then, by our charging scheme, 
we have $\left| V_d\right| = 1$, and $d$ is charged at most twice.
Otherwise, if $d$ is assigned to a lightly-loaded vertex,
then, by our charging scheme, each vertex in $V_d$ charges $d$ at most once, disregarding heavily-loaded
or lightly-loaded vertices.

By the above argument and Lemma~\ref{lemma_bounded_degree_reduction}, 
we have $w(f^*) = \sum_{u \in V}w(u)x^*(u)\le \sum_{u \in V}d(u) \cdot deg[u]\cdot y_u 
\le 7\cdot\sum_{u\in V}d(u)\cdot y_u \le 7\cdot OPT_f(\HH_i)$,
where the last inequality follows from the linear program duality of $(\ref{ILP_cd})$ and $(\ref{ILP_dual_cd})$.
\end{proof}

\begin{figure}[h]
\centering
\includegraphics[scale=0.8]{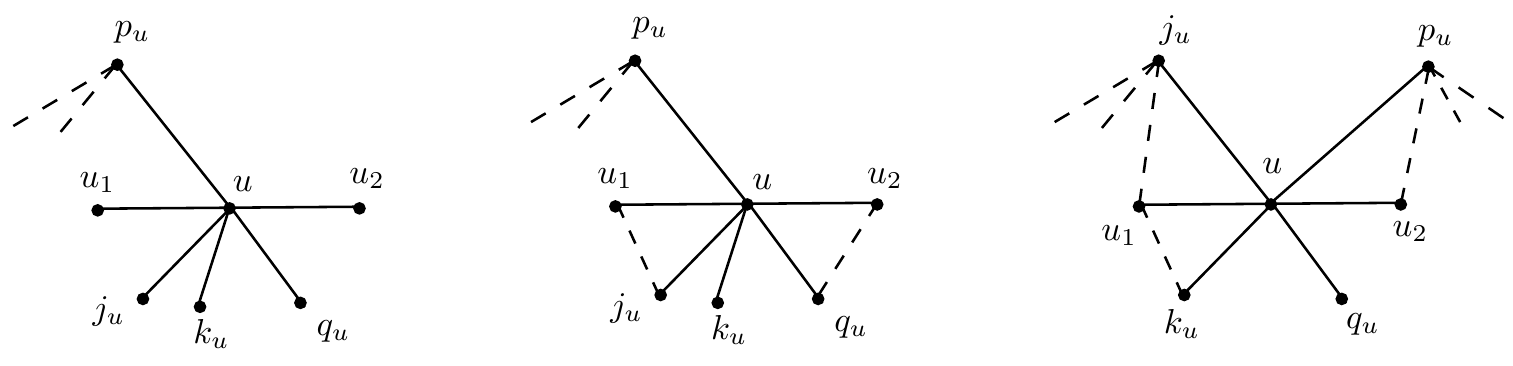}
\vspace{-18pt}
\caption{Situations when a unit demand of $u$ is fully-charged.}
\label{figure_improved_charging_}
\end{figure}

Thanks to the structural property provided in Lemma~\ref{lemma_ladder_upward_degree_bound}, 
given the fact that the input graph
is outer-planar, we can modify the algorithm slightly and further improve
the bound given in the previous lemma.
To this end, we consider the situations when a unit demand from a vertex $u$ with $deg[u] = 7$
and argue that, either it is not fully-charged by all its closed neighbours, or
we can modify the demand assignment, without raising the cost, to make it so.
%

\begin{lemma} \label{lemma_primal_dual_outerplanar_ladder}
Given the fact that $\HH_i$ comes from an outerplanar graph, we can modify the algorithm
to obtain a demand assignment function $f^*$ such that $w(f^*) \le 6\cdot OPT_f(\HH_i)$.
\end{lemma}

\begin{proof}[Proof of Lemma~\ref{lemma_primal_dual_outerplanar_ladder}]
Consider any unit demand, say demand $d$ from vertex $u$ in $\LL_j$, and let $V_d \subseteq N[u]$ be the set of
vertices that has charged $d$ by our charging scheme.


First, we have $\BIGLR{|}{|}{N[u]} \le 7$ by Lemma~\ref{lemma_bounded_degree_reduction}.
By our charging scheme, $\BIGLR{|}{|}{N[u]} < 7$ implies $\BIGLR{|}{|}{V_d} < 7$.
%
In the following, we assume 
$\BIGLR{|}{|}{N[u]}=7$ and argue that either we have $\BIGLR{|}{|}{V_d} < 7$,
or we can modify the solution such that $\BIGLR{|}{|}{V_d} < 7$.
By assumption, $k_u$, $p_u$, $q_u$ are well-defined.
Let $u_1, u_2 \in N(u)\cap \LL_j$ denote the set of neighbours of $u$ in $\LL_j$ such that $u_1 \prec u \prec u_2$.
By Lemma~\ref{lemma_ladder_upward_degree_bound},
depending on the layer to which $j_u$ and $k_u$ belong, we have the following two cases.

\paragraph{Both $j_u$ and $k_u$ belong to $\LL_{j+1}$.}

If two of $\{j_u, k_u, q_u\}$, say, $j_u$ and $k_u$, are not joined to $u_1$ and $u_2$ by any edge,
then at most one of $j_u$ and $k_u$ can charge $d$, since $u$ is the only vertex with possibly non-zero demand in
their closed neighbourhoods.
When the first closed neighbour of $u$ is saturated and $u$ is removed from $V^\phi$,
both $j_u$ and $k_u$ will be removed from $V^*$ and will not be picked in later iterations.
Therefore, at most one of $j_u$ and $k_u$ can charge $d$.

On the other hand, if two of $\{j_u, k_u, q_u\}$, say $j_u$ and $q_u$, are joined to $u_1$ and $u_2$, respectively,
then we argue that at most two out of $\{j_u, u, q_u\}$ can charge $d$.
Indeed, $u_1$, $u$, and $u_2$ are the only vertices with non-zero demands in the closed neighbourhoods of
$\{j_u, u, q_u\}$. After two of $\{j_u, u, q_u\}$ is saturated, $u_1$, $u$, and $u_2$ will
be removed from $V^\phi$. Therefore, at most two out of $\{j_u, u, q_u\}$ can charge $d$.
See also Fig.~\ref{figure_improved_charging_}~(a) and (b).

\paragraph{Only one of $\{j_u, k_u\}$ belongs to $\LL_{j+1}$ and the other belongs to $\LL_{j-1}$.}

Without loss of generality, we assume that $j_u \in \LL_{j-1}$ and $k_u \in \LL_{j+1}$.
By Lemma~\ref{lemma_ladder_upward_degree_bound}, both $j_u$ and $p_u$ are joined either to $u_1$
or $u_2$ separately. Since $\HH_i$ is outerplanar, we have $j_u \prec u \prec p_u$, otherwise $j_u$ will
be contained inside the face surrounded by $p_u$, $u_1$, and $u$, which is a contradiction.
See also Fig.~\ref{figure_improved_charging_}~(c).

If both $k_u$ and $q_u$ are not joined to either $u_1$ or $u_2$, then by a similar argument we used
in previous case, at most one of $k_u$ and $q_u$ can charge $d$.
Now, suppose that, one of $\{k_u, q_u\}$, say, $k_u$, is joined to $u_1$ by an edge.
We argue that, if both $u_1$ and $k_u$ have charged $d$ after $d$ has been assigned in a feasible solution returned by
our algorithm, then we can cancel the multiplicity placed on $k_u$ and
reassign to $u_1$ the demand which was previously assigned to $k_u$ without increasing the cost spent on $u_1$.

If $u_1$ is lightly-loaded in the beginning, then the above operation can be done without extra cost.
Otherwise, observe that, in this case, $u_1$ must have been lightly-loaded when $d$ is assigned
so that it can charge $d$ later. Moreover, $u_1$ is also in the set $V^\phi$ and not yet served, for otherwise
$k_u$ will be removed from $V^*$ and will not be picked later.
In other words, at this moment when $u_1$ becomes lightly-loaded, we have $u_1 \in V^\phi$,
meaning that it is possible to assign the demand of $u_1$ to itself later without extra cost.

On the other hand, if $u_1$ or $k_u$ is the first one to charge $d$, consider the relation between
$q_u$ and $u_2$. If there is no edge between $q_u$ and $u_2$, then $q_u$ will not be picked and will
not charge $d$. Otherwise, if $(q_u,u_2)$ exists in $\HH_i$, then it is a symmetric situation described in the
last sequel.

In both cases, either we have $\BIGLR{|}{|}{V_d} \le 6$ or we can modify the solution returned by the algorithm
to make $\BIGLR{|}{|}{V_d} \le 6$ to hold.
Therefore we have $w(f^*) \le 6\cdot OPT_f(\HH_i)$ as claimed.
\end{proof}

\subsection{Overall Analysis} \label{section_overall_analysis}

We summarize the whole algorithm and our main theorem.
Given an outer-planar graph $G=(V,E)$, we use the algorithm described in \S\ref{subsection_structure}
to compute a general-ladder representation of $G$, followed by applying the decomposition
to obtain three subproblems, $\GG_0$, $\GG_1$, and $\GG_2$.
For each $\GG_i$,
we use the approach described in \S\ref{subsection_reduction} to further
remove more edges and obtain the reduced subgraph $\HH_i$, for which we apply the algorithm described
in \S\ref{section_greedy_charging} to obtain an approximation, which is a demand assignment
function $f_i$ for $\HH_i$. The overall approximation, e.g., the demand assignment function $f$, 
for $G$ is defined as $f = \sum_{0\le i<3}f_i$.

\begin{theorem} \label{thm_overall_analysis}
Given an outerplanar graph $G$ as an instance of capacitated domination, 
we can compute a constant factor approximation for $G$ in $O(n^2)$ time.
\end{theorem}

\begin{proof}[Proof of Theorem~\ref{thm_overall_analysis}]
First, we argue that the procedures we describe can be done in $O(n^2)$ time.
It takes $O(n\log^3n)$ time to compute an outer-planar representation of $G$~\cite{Bose97onembedding}.
By Theorem~\ref{theorem_computing_representation}, computing a general-ladder representation takes linear time.
The construction of $\GG_i$ takes time linear in the number of edges, which is linear in $n$ since
$G$ is planar.


In the construction of the reduced graphs $\HH_i$, for each vertex $v$, 
although we use a sorted list of the closed neighbourhood of $v$ to define $j_v$, $k_v$, $p_v$, and $q_v$,
the sorted lists are not necessary and $\HH_i$ can be constructed in $O(n)$ time by a careful implementation.
This is done in a two-passes traversal on the set of edges of $\GG_i$ as follows.
In the first pass, we iterate over the set of edges to locate $j_v$, $p_v$, and $q_v$, for each vertex $v\in V$.
Specifically, we keep a current candidate for each vertex and for each edge $(u,v) \in E$ iterated,
we make an update on $u$ and $v$ if necessary.
In the second pass, based on the $j_v$ computed for each $v \in V$, we iterate over the set of edges again
to locate $k_v$.
The whole process takes time linear in the number of edges, which is $O(n)$.


In the following, we explain how the primal-dual algorithm, i.e., the algorithm presented in
Figure~\ref{algorithm_greedy_charging}, can be implemented to
compute a feasible solution in $O(n^2)$ time.
First, we traverse the set of edges in linear time to compute the value $d^\phi(v)$ for each vertex.
In each iteration, the next vertex to be saturated, which is the one with
minimum $w^\phi(v) / \min\BIGLR{\{}{\}}{c(v), d^\phi(v)}$, can be found in linear time.
The update of $w^\phi(v)$ for each $v \in V^*$ described in line $8$ can be done in linear time.
When a vertex $v \in N[u]$ with non-zero demand is removed from $V^\phi$, we have to update the value $d^\phi(v^\prime)$
for all $v^\prime \in N[v]$. By Lemma~\ref{lemma_bounded_degree_reduction}, the closed degree of such vertices is bounded by $7$.
This update can be done in $O(1)$ time.
The construction of $\SSS_u$ can be done in linear time.
Since $d^\phi(v)$ can only decrease, each vertex can turn into lightly-loaded at most once.
Therefore the process time for these vertices is bounded in linear time.
The outer-loop iterates at most $O(n)$ times. Therefore the whole algorithm runs in $O(n^2)$ time.


The feasibility of the demand assignment function $f$ is guaranteed by Lemma~\ref{lemma_layers_decomp_quality}
and the fact that $\HH_i$ is a subgraph of $\GG_i$.
Since $\HH_i \subseteq \GG_i$, the demand assignment we obtained for $\HH_i$ is also a feasible
demand assignment for $\GG_i$. Therefore, $f$ is feasible for $G$.


By the definition of $f$, Lemma~\ref{lemma_primal_dual_outerplanar_ladder}, Lemma~\ref{lemma_bounded_degree_reduction}, 
and Lemma~\ref{lemma_layers_decomp_quality},
we have
\begin{align*}
w(f) & \le \sum_{0\le i<3}w(f_i) \quad  \le 6\cdot \sum_{0\le i<3}OPT_f(\HH_i) \le 12\cdot \sum_{0\le i<3}OPT_f(\GG_i) \\
 & \le 12\cdot\sum_{0\le i<3}OPT(\GG_i) \le 36\cdot OPT(G).
\end{align*}

\end{proof}

\subsection{Extension to Planar Graphs} \label{section_extension_planar}

We describe how our outer-planar result can be extended to obtain a constant factor approximation for planar graphs
under a general framework due to~\cite{174650}.
This is done as follows. Given a planar graph $G$, we generate a planar embedding and retrieve the vertices of each level 
using the linear-time algorithm of Hopcroft and Tarjan \cite{321852}. 

%
\begin{figure}[h]
\centering\fbox{\begin{minipage}{0.8\textwidth}
\centering
\includegraphics[scale=0.7]{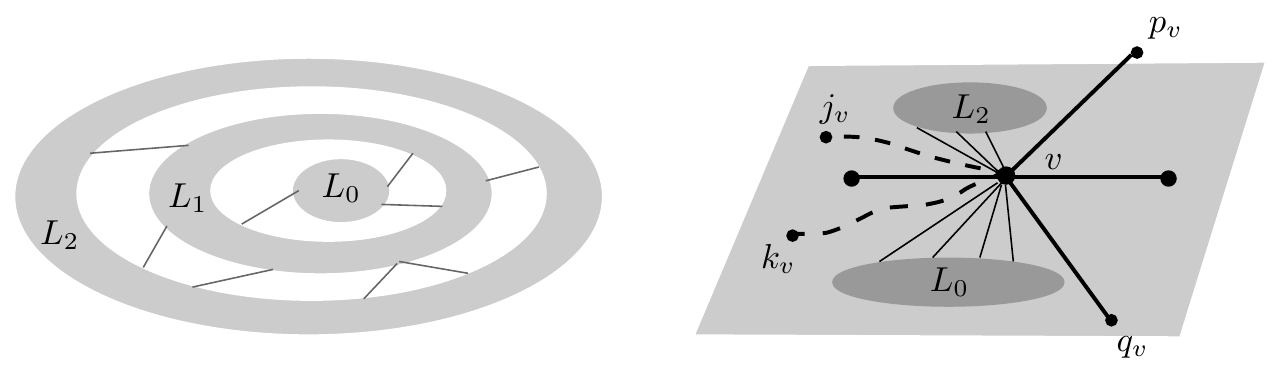}
\caption{(a) $3$-outerplanar graph.
(b) Local connections with respect to a vertex $v$.
Bold edges represent links in the ladder extracted from $L_1$. Thin edges represent
links between $L_1$ and the other two levels, $L_0$ and $L_2$.}
\label{figure_planar_extension_connection_}
\end{minipage}}
\end{figure}
%

Let $m$ be the number of levels of this embedding. 
Let $OPT$ be the cost of the optimal demand assignment of $G$, 
and $OPT_j$ be the cost contributed by vertices at level $j$.
For convenience, in the following, for $j\le 0$ or $j>m$, 
we refer the vertices in level $j$ to an empty set and the corresponding cost $OPT_j$ is defined to be zero.

For $0\le i<3$, we define $C_i$ as $\sum_{0\le j\le \frac{m}{3}}\BIGP{OPT_{3\cdot j-i}+OPT_{3\cdot j-i+2}}$.
Since
$\sum_{0\le i<3}C_i \le 2\cdot OPT$, there exist an $i_0$ with $0\le i_0<2$ such that
$C_{i_0} \le \frac{2}{3}\cdot OPT$.
For each $0\le j\le \frac{m}{3}$, define the graph $G_j$ to be the graph induced by vertices between
level $3\cdot j-i_0$ and level $3\cdot j-i_0+2$. 
The parameters of the vertices in $G_j$ are set as follows.
For those vertices who are from level $3\cdot j-i_0$ and level $3\cdot j-i_0+2$,
their demands are set to be zero. The rest parameters are remained unchanged.
Clearly, $G_j$ is a $3$-outerplanar graph,
and we have $\sum_{0\le j\le \frac{m}{3}}OPT(G_j) \le OPT+\frac{2}{3}\cdot OPT$,
where $OPT(G_j)$ is the cost of the optimal demand assignment of $G_j$.


In the following, we sketch how our algorithm for outerplanar graphs can be modified slightly and
applied to $G_j$ for each $0\le j\le \frac{m}{3}$ to obtain a constant approximation for $G_j$.
For convenience, we denote the set of vertices from the three levels of $G_j$ by $L_0$, $L_1$, and $L_2$, respectively.
See Fig.~\ref{figure_planar_extension_connection_}~(a).

\begin{itemize}
\item{\bf Obtaining the General Ladders.}
For each level $L_i$, $0\le i<3$, which constitutes an outerplanr graph by itself, we define a total order over it according to
the order of appearances of the vertices
in counter-clockwise order. The general-ladder is extracted from $L_1$ as we did before.
Furthermore, for each vertex in the ladder, its incident edges to vertices in $L_0$ and $L_2$ are also included.

\item{\bf Removing Redundant Edges.}
In addition to the four vertices we identified for each vertex $v$ with non-zero demand,
we identify two more vertices, which literally corresponds to the rightmost neighbours of $v$ in level $L_0$ and
$L_2$, respectively. See also Fig.~\ref{figure_planar_extension_connection_}~(b).
Then, the first part of Lemma~\ref{lemma_bounded_degree_reduction} still holds, 
the degree bound provided in the second part is increased by $2$,
and the third part holds automatically without alteration.
\end{itemize}

The rest parts of our algorithm remain unchanged. From the above argument, 
we have the following theorem.

\bigskip

\begin{theorem} \label{theorem_extension_planar}
Given a planar graph $G$ as an instance of capacitated domination, we can compute a constant factor approximation
for the $G$ in polynomial time.
\end{theorem}


\section{Conclusion} \label{conclusion}

%
The results we provide
seem to have room for further improvements.
One reason is that, due to the flexibility of the ways the demand can be assigned,
it seems not promising to come up with an approximation threshold.
%
%
However, when the demand cannot be split, it is not difficult to 
prove a constant approximation threshold.
Therefore, it would be very interesting to investigate the problem complexity
on planar graphs.


Second, 
as we have shown in \S\ref{subsection_structure}, the concept of general-ladders does not extend
directly to $k$-outerplanar graphs for $k \ge 2$. 
It would be interesting to formalize and extend this concept to $k$-outerplanar graphs,
for it seems helpful not only to our problem, but also to most capacitated covering problems as well.


\subsubsection*{Acknowledgements.} The authors would like to thank the anonymous referees 
for their very helpful comments on the layout of this work.


\bibliographystyle{siam}

\small
\bibliography{approx_capacitated_domination}

\end{document}